\documentclass[11pt]{article}
\usepackage{makeidx}
\usepackage{color}
\usepackage{graphicx}
\usepackage{fullpage}
\usepackage{latexsym,amssymb}

\newtheorem{definition}{Definition}
\newtheorem{theorem}{Theorem}
\newenvironment{proof}{\noindent\textbf{Proof:}}{\hfill$\Box$}

\begin{document}
%
%\frontmatter          % for the preliminaries
%
%\pagestyle{headings}  % switches on printing of running heads
%\addtocmark{Hamiltonian Mechanics} % additional mark in the TOC
%

\title{On Byzantine Broadcast in \\ Loosely Connected Networks}
\author{Alexandre Maurer$^1$ \and S\'{e}bastien Tixeuil$^{1,2}$}
\date{\small
$^1$ UPMC Sorbonne Universit\'{e}s, LIP6, LINCS, France\\
$^2$ Institut Universitaire de France\\
\texttt{Alexandre.Maurer@lip6.fr}\\
\texttt{Sebastien.Tixeuil@lip6.fr}
}

\maketitle              % typeset the title of the contribution

\begin{abstract}

We consider the problem of reliably broadcasting information in a multihop asynchronous network that is subject to Byzantine failures. Most existing approaches give conditions for perfect reliable broadcast (all correct nodes deliver the authentic message and nothing else), but they require a highly connected network. An approach giving only probabilistic guarantees (correct nodes deliver the authentic message with high probability) was recently proposed for loosely connected networks, such as grids and tori. Yet, the proposed solution requires a specific initialization (that includes global knowledge) of each node, which may be difficult or impossible to guarantee in self-organizing networks -- for instance, a wireless sensor network, especially if they are prone to Byzantine failures.

In this paper, we propose a new protocol offering guarantees for loosely connected networks that does not require such global knowledge dependent initialization. In more details, we give a methodology to determine whether a set of nodes will always deliver the authentic message, in any execution. Then, we give conditions for perfect reliable broadcast in a torus network. Finally, we provide experimental evaluation for our solution, and determine the number of randomly distributed Byzantine failures than can be tolerated, for a given correct broadcast probability.

\end{abstract}

\section{Introduction}

In this paper, we study the problem of reliably broadcasting information in a network that is subject to attacks or failures. Those are an important issue in a context where networks grow larger and larger, making the possibility of failure occurrences more likely. Many models of failures and attacks have been studied so far, but the most general model is the \emph{Byzantine} model \cite{LSP82j}: some nodes in the network may exhibit arbitrary behavior. In other words, all possible behaviors must be anticipated, including the most malicious strategies. The generality of this model encompasses a rich panel of security applications.

In the following, we assume that a correct node (the \emph{source}) broadcasts a message in a network that may contain Byzantine nodes. We say that a correct node \emph{delivers} a message, when it considers that this actually is the message broadcasted by the source.

\paragraph{Related works.}

Many Byzantine-robust protocols are based on \emph{cryptography} \cite{CL99c,DFS05c}: the nodes use digital signatures or certificates. Therefore, the correct nodes can verify the validity of received informations and authenticate the sender across multiple hops. However, this approach weakens the power of Byzantine nodes, as they ignore some cryptographic secrets: their behavior is not totally arbitrary.
Moreover, in some applications such as sensor networks, the nodes may not have enough resources to manipulate digital signatures.
Finally, cryptographic operations require the presence of a trusted infrastructure, such as secure channels to a key server or a public key infrastructure.
In this paper, we focus on non-cryptographic and totally distributed solutions: no element of the network is more important than another, and all elements are likely to fail.

Cryptography-free solutions have first been studied in completely connected networks~\cite{LSP82j,AW98b,MMR03j,MRRS01c,MS03j}: a node can directly communicate with any other node, which implies the presence of a channel between each pair of nodes. Therefore, these approaches are hardly scalable, as the number of channels per node can be physically limited. We thus study solutions in partially connected networks, where a node must rely on other nodes to broadcast informations.

Dolev~\cite{D82j} considers Byzantine agreement on arbitrary graphs, and states that for agreement in the presence of up to $k$ Byzantine nodes, it is necessary and sufficient that the network is $(2k+1)$-connected and the number of nodes in the system is at least $3k+1$. Also, this solution assumes that the topology is known to every node, and that nodes are scheduled according to the synchronous execution model. Nesterenko and Tixeuil~\cite{NT09j} relax both requirements (the topology is unknown and the scheduling is asynchronous) yet retain $2k+1$ connectivity for resilience and $k+1$ connectivity for detection (the nodes are aware of the presence of a Byzantine failure). In sparse networks such as a grid (where a node has at most four neighbors), both approaches can cope only with a single Byzantine node, independently of the size of the grid. More precisely, if there are two ore more Byzantine nodes \emph{anywhere} in the grid, there always exists a possible execution where no correct node delivers the authentic message.

Byzantine-resilient broadcast was also investigated in the context of \emph{radio networks}: each node is a robot or a sensor with a physical position. A node can only communicate with nodes that are located within a certain radius. Broadcast protocols have been proposed \cite{K04c,BV05c} for nodes organized on a grid. However, the wireless medium typically induces much more than four neighbors per node, otherwise the broadcast does not work (even if all nodes are correct). Both approaches are based on a local voting system, and perform correctly if every node has less than a $1/4\pi$ fraction of Byzantine neighbors. This criterion was later generalized \cite{PP05j} to other topologies, assuming that each node knows the global topology. Again, in loosely connected networks, the local constraint on the proportion of Byzantine nodes in any neighborhood may be difficult to assess.

A notable class of algorithms tolerates Byzantine failures with either space~\cite{MT07j,NA02c,SOM05c} or time~\cite{MT06cb,DMT11cb,DMT11j,DMT10cd,DMT10ca} locality. Yet, the emphasis of space local algorithms is on containing the fault as close to its source as possible. This is only applicable to the problems where the information from remote nodes is unimportant (such as vertex coloring, link coloring or dining philosophers). Also, time local algorithms presented so far can hold at most one Byzantine node and are not able to mask the effect of Byzantine actions. Thus, the local containment approach is not applicable to reliable broadcast.

All aforementioned results rely on strong \emph{connectivity} and Byzantine proportions assumptions in the network. In other words, tolerating more Byzantine failures requires to increase the connectivity, which can be a heavy constraint in a large network. To overcome this problem, a probabilistic approach for reliable broadcast has been proposed in \cite{CtrZ}. In this setting, the distribution of Byzantine failures is assumed to be random. This hypothesis is realistic in various networks such as a peer-to-peer overlays, where the nodes joining the network are not able to choose their localization, and receive a randomly generated identifier that determines their location in the overlay. Also, it is considered acceptable that a small minority of correct nodes are fooled by the Byzantine nodes. With these assumptions, the network can tolerate~\cite{CtrZ} a number of Byzantine failures that largely exceeds its connectivity. Nevertheless, this solution requires to define many sets of nodes (called \emph{control zones}~\cite{CtrZ}) before running the protocol: each node must initially know to which control zones it belongs. This may be difficult or impossible in certains types of networks, such as a self-organized wireless sensor network or a peer-to-peer overlay. 

\paragraph{Our contribution.}

In this paper, we propose a broadcast protocol performing in loosely connected networks subject to Byzantine failures that relaxes the aforementioned constraint -- no specific initialization is required for the nodes. This protocol is described in Section~\ref{sec_desc}. Further, we prove general properties on this protocol, and use them to give both deterministic and probabilistic guarantees.

In Section~\ref{sec_prop}, we give a sufficient condition for \emph{safety} (no correct node delivers a false message). This condition is not based on the number, but on the \emph{distance} (with respect to the number of hops) between Byzantine failures. Then, we give a methodology to construct -- node by node -- a set of correct nodes that will always deliver the authentic message, in any possible execution.

In Section~\ref{sec_torus}, we consider a particular loosely connected network: the \emph{torus}, where each node has exactly four neighbors. We give a sufficient condition to achieve perfect reliable broadcast on such a network (all correct nodes deliver the authentic message).

In Section~\ref{sec_exp}, we make an experimental evaluation of the protocol on \emph{grid} networks. We give a methodology to estimate the probability that a correct node delivers the authentic message, for a given number of Byzantine failures. This way, we can determine the maximal number of failures that the network can hold, to achieve a given probabilistic guarantee.

\section{Description of the protocol}
\label{sec_desc}

In this section, we provide an informal description of the protocol. Then, we precise our notations and hypotheses, and give the algorithm that each correct node must follow.

\subsection{Informal description}

The network is described by a set of processes, called \emph{nodes}. Some pairs of nodes are linked by a \emph{channel}, and can send messages to each other: we call them \emph{neighbors}. The network is \emph{asynchronous}: the nodes can send and receive messages at any time.

A particular node, called the \emph{source}, wants to broadcast an information $m$ to the rest of the network. In the ideal case, the source would send $m$ to its neighbors, which will transmit $m$ to their own neighbors -- and so forth, until every node receives $m$. In our setting however, some nodes -- except the source --- can be malicious (\emph{Byzantine}) and broadcast false informations to  the network. Of course, a correct node cannot know whether a neighbor is Byzantine.

To limit the diffusion of false messages, we introduce a \emph{trigger} mechanism: when a node $p$ receives a message $m$, it must wait the reception of a \emph{trigger message} to accept and retransmit $m$. The \emph{trigger message} informs $p$ that another node, located at a reasonable distance, has already accepted $m$. This distance is the number $H$ of channels (or \emph{hops}) that the trigger message can cross. This is illustrated in Figure~\ref{fig:trigger}-a.

The underlying idea is as follows: if the Byzantine nodes are sufficiently spaced, they will never manage to broadcast false messages. Indeed, to broadcast a false message, a Byzantine node requires an accomplice to broadcast the corresponding trigger message (see Figure~\ref{fig:trigger}-b). However, if this accomplice is distant from more than $H+1$ hops, the trigger message will never reach its target, and the false message will never be accepted (see Figure~\ref{fig:trigger}-c).

\begin{figure}
\begin{center}
\includegraphics[width=\textwidth]{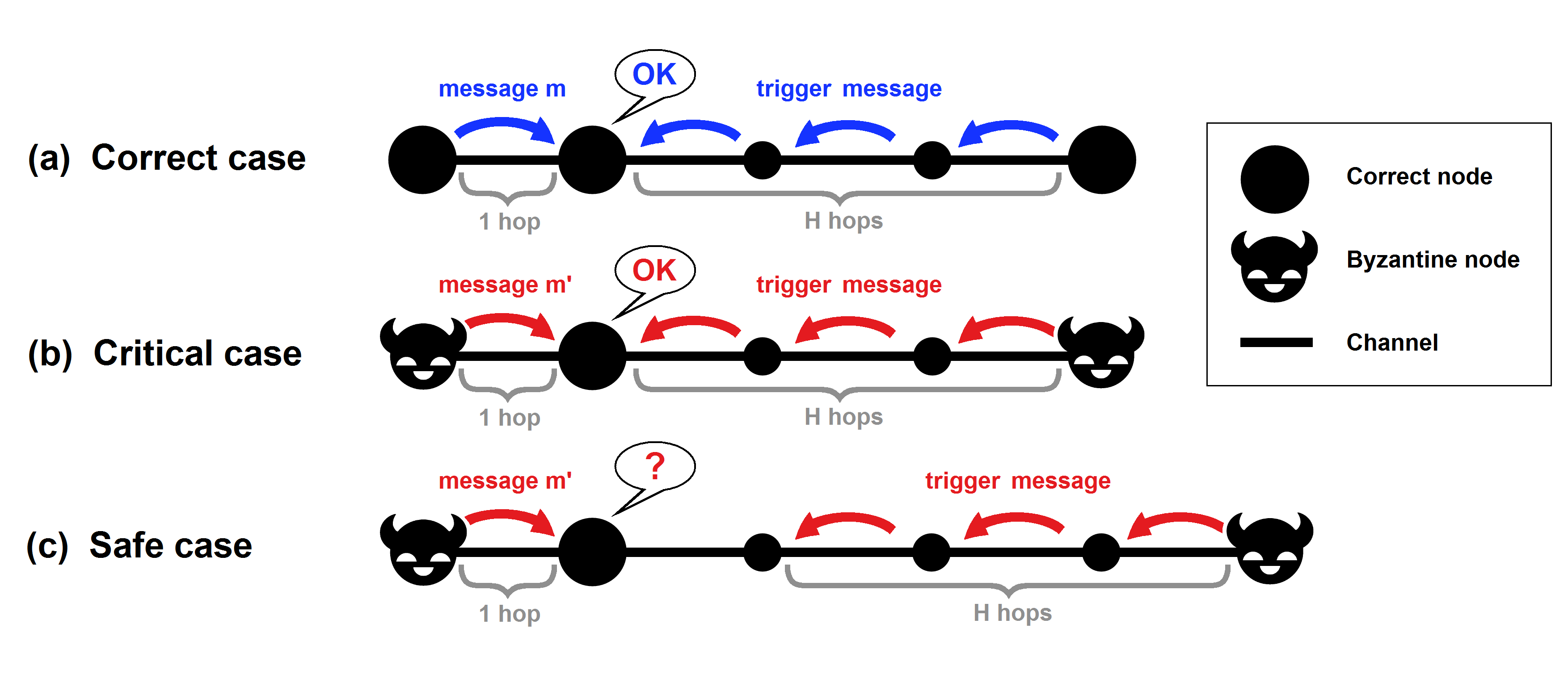}
\caption{Illustration of the trigger mechanism} 
\label{fig:trigger}
\end{center}
\end{figure}

\subsection{Notations and Hypotheses}

Let $(G,E)$ be a non-oriented graph representing the topology of the network. $G$ denotes the \emph{nodes} of the network. $E$ denotes the \emph{neighborhood} relationship. A node can only send messages to its neighbors. Some nodes are $correct$ and follow the protocol described thereafter. We consider that all other nodes are totally unpredictable (or \emph{Byzantine}) and may exhibit an arbitrary behavior.

\paragraph{Hypotheses}

We consider an asynchronous message passing network: any message sent is eventually received, but it can be at any time. We assume that, in an infinite execution, any process is activated inifinitely often. However, we make no hypothesis on the order of activation of the processes. Finally, we assume local topology knowledge: when a node receives a message from a neighbor $p$, it knows that $p$ is the author of the message. Therefore, a Byzantine node cannot lie about its identity to its direct neighbors. This model is referred to as the ``oral'' model in the literature (or \emph{authenticated channels}).

\paragraph{Messages formalism}

In the protocol, two types of messages can be exchanged:
\begin{itemize}
\item \emph{Standard messages}, of the form $(m)$: a message claiming that the source broadcasted the information $m$.
\item \emph{Trigger messages}, of the form $(m,S)$: a message claiming that a node has delivered $m$. The set $S$ should contain the identifiers of the nodes visited by this message.

\end{itemize}

The protocol is characterized by a parameter $H \geq 1$: the maximal number of hops that a trigger message can cross. Typically, this limit is reached when $S$ contains more than $H-1$ nodes. This parameter is known by all correct nodes.

\paragraph{Local memories}

Each correct node $p$ maintains two dynamic sets, initially empty:

\begin{itemize}
\item  $Wait$: the set of standard messages received, but not yet accepted. When $(m,q) \in Wait$, it means that $p$ received a standard message $(m)$ from a neighbor $q$.
\item $Trig$: set of trigger messages received. When $(m,S) \in Trig$, it means that $p$ received a trigger message $(m,S-\{q\})$ from a neighbor $q$.
\end{itemize}

\paragraph{Vocabulary}

We will say that a node $multicasts$ a message when it sends it to all its neighbors.
A node $delivers$ a message $m$ when its consider that it is the authentic information broadcast by the source.
In the remaining of the paper, we call $D$ the shortest number of hops between two Byzantine nodes. For instance, $D = 4$ in Figure~\ref{fig:trigger}-b, and $D = 5$ in Figure~\ref{fig:trigger}-c.

\label{prelim_proto}

\subsection{Local Execution of the Protocol}

\label{localex}

Initially, the source multicasts $m$ and $(m,\o)$. Then, each correct node follows these three rules:

\begin{itemize}
\item RECEPTION -- When a standard message $(m)$ is received from a neighbor $q$: if $q$ is the source, deliver $m$, then multicast $(m)$ and $(m,\o)$; else, add $(m,q)$ to the set $Wait$.

\item TRANSMISSION -- When a trigger message $(m,S)$ is received from a neighbor $q$: if $q \notin S$ and $card(S) \leq H-1$, add $(m,S \cup \{q\})$ to the set $Trig$ and multicast $(m,S \cup \{q\})$.

\item DECISION -- When there exists $(m,q,S)$ such that $(m,q) \in Wait$, $(m,S) \in Trig$ and $q \notin S$: deliver $m$, then multicast $(m)$ and $(m,\o)$.
\end{itemize}

\section{Protocol Properties}

\label{sec_prop}

In this section, we give conditions about the placement of Byzantine nodes that guarantee network \emph{safety} (that is, no correct node ever delivers a false message). Then, we give a methodology to compute a set of nodes that always delivers authentic messages, in any possible execution. Remind that correct nodes do \emph{not} know the actual positions of Byzantine nodes.

\subsection{Network Safety}
\label{sec:safety}

The following theorem guarantees network safety, provided that Byzantine node are sufficiently spaced. This condition depends on the parameter $H$ of the protocol, and on the distance $D$ (see \ref{prelim_proto}). We also show that the condition on $D$ is tight for our protocol.

Notice that safety does not guarantee that correct nodes actually \emph{deliver} the authentic message. This aspect is studied in \ref{sec:reliability}.

\begin{theorem}[Network Safety]
\label{thsafe}
If $D \geq H+2$, no correct nodes delivers a false message.
\end{theorem}

\begin{proof}  The proof is by contradiction. Let us suppose the opposite : $D \geq H+2$, and at least one correct node delivers a false message.
Let $u$ be the first correct node to deliver a false message, and let $m'$ be this message.

No correct node can deliver $m'$ in RECEPTION, as the source did not send $m'$. So $u$ delivered $m'$ in DECISION, implying that there exists $q$ and $S$ such that $(m',q) \in u.Wait$, $(m',S) \in u.Trig$ and $q \notin S$.

The statement $(m',q) \in u.Wait$ implies that $u$ received $(m')$ from a neighbor $q$ in RECEPTION. Let us suppose that $q$ is correct. Then, $q$ sent $(m')$ in DECISION, implying that $q$ delivered $m'$. This is impossible, as $u$ is the first correct node to deliver $m'$. So $q$ is necessarily Byzantine.

Now, let us prove the following property $\mathcal{P}_i$ by recursion,
for $1 \leq i \leq H+1$: a correct node $u_i$, at $i$ hops or less from $q$, received a message $(m',S_i)$, and $card(S) = card(S_i) + i$.

\begin{itemize}
\item First, let us show that $\mathcal{P}_1$ is true.
The statement $(m',S) \in u.Trig$ implies that $u$ received $(m',\mathcal{X})$ from a neighbor $x$ in TRANSMISSION, with
$S = \mathcal{X} \cup \{x\}$
and
$x \notin \mathcal{X}$,
So $card(S) = card(\mathcal{X}) + 1$.
Therefore, $\mathcal{P}_1$ is true if we take $u_1 = u$ and $S_1 = \mathcal{X}$.
Besides, it is also necessary that $card(\mathcal{X}) \leq H - 1$, so $card(S) \leq H$.

\item Let us suppose that $\mathcal{P}_i$ is true, with $i \leq H$.
The node $u_i$ received $(m',S_i)$ from a node $x$, so $x$ is at $i+1$ hops or less from $q$.
Let us suppose that $x$ is Byzantine. Then, according to the previous statement, $D \leq i+1 \leq H+1$, contradicting our hypothesis. So $x$ is necessarily correct.

Node $x$ could not have sent $(m',S_i)$ in RECEPTION or DECISION, as $u$ is the first correct node to deliver $m'$. 
So this happened in TRANSMISSION, implying that $x$ received $(m',\mathcal{Y})$ from a node $y$, with $S_i = \mathcal{Y} \cup \{y\}$ and $y \notin \mathcal{Y}$.
So $card(S_i) = card(\mathcal{Y}) + 1$, and $card(S) = card(\mathcal{Y}) + i + 1$.
Therefore, $\mathcal{P}_{i+1}$ is true if we take $u_{i+1} = x$ and $S_{i+1} = \mathcal{Y}$.
\end{itemize}
Overall, $\mathcal{P}_{H+1}$ is true and $card(S) = card(S_{H+1}) + H+1 \geq H+1$. But, according to a previous statement, $card(S) \leq H$. This contradiction completes the proof.
\end{proof}

As a complementary result, let us show that the bound $D \geq H + 2$ is tight for our protocol.

\begin{theorem} [Tight bounds for safety]
\label{thunsafe}
If $D = H + 1$, some correct nodes may deliver a false message.
\end{theorem}

\begin{proof}
Let $b$ and $c$ be two Byzantine nodes distant from $H + 1$ hops. Let $(p_0,...,p_{H+1})$ be a path of $H + 1$ hops, with $p_0 = b$ and $p_{H+1} = c$.
Then, $b$ can send a standard message $(m')$ to $p_1$, and $c$ can send the trigger message for $m'$ trough $H$ hops. Therefore, it is possible that $p_1$ delivers the false message, and the network is not safe.
\end{proof}

\subsection{Network Reliability}
\label{sec:reliability}

Here, we suppose that the safety conditions determined in Section~\ref{sec:safety} are satisfied: no correct node can deliver a false message. We now give a methodology to construct a set $S$ of nodes that always delivers the authentic message. 

\begin{definition}[Reliable node set]
\label{defrel}
For a given source node and a given distribution of Byzantine nodes, a set of correct nodes $S$ is \emph{reliable} if all nodes in $S$ eventually deliver authentic messages in any possible execution.
\end{definition}

\begin{definition}[Correct path]
A $N$-hops \emph{correct path} is a sequence of distinct correct nodes $(p_0, \dots , p_N)$ such that, $\forall i \leq N-1$, $p_i$ and $p_{i+1}$ are neighbors.
\end{definition}

Notice that, according to RECEPTION (see \ref{localex}), the set formed by the source and its correct neighbors is reliable. The following theorem permits to decide whether a given node $p$ can be added to a reliable set $S$. So, a reliable set can be extended node by node, and can potentially contain the majority or the totality of the correct nodes.

\begin{theorem} [Reliable set determination]
Let us suppose that the hypotheses of Theorem~\ref{thsafe} (Network Safety) are all satisfied.
Let $S$ be a reliable node set, and $p \notin S$ a node with a neighbor $q \in S$.
If there exists a correct path of $H$ hops or less between $p$ and a node $v \in S$ (all nodes of the path being distinct from $q$), then $S \cup \{p\}$ is also a reliable node set.
\label{thcom}
\end{theorem}

\begin{proof}
Let $m$ be the message broadcast by the source. As the hypotheses of Theorem~\ref{thsafe} are satisfied, the correct nodes can only deliver $m$. 
As $q$ and $v$ are in a reliable node set, there exists a configuration where $q$ and $v$ have delivered $m$.
This implies that $q$ and $v$ have multicast $(m)$ and $(m,\o)$.

So $p$ eventually receives $(m)$ from $q$.
If $q$ is the source, $p$ delivers $m$, completing the proof.
Now, let us suppose that $q$ is not the source.
Then, $p$ eventually adds $(m,q)$ to its set $Wait$ in RECEPTION.

Let $(v_0,\dots,v_N)$ be a $N$-hops correct path, with $v_0 = v$, $v_N = p$ and $N \leq H$. Let $S_i$ be the set of nodes defined by $S_0 = \o$ and $S_i = \{v_0,\dots,v_{i-1}\}$ for $1 \leq i \leq N$.
Let us prove the following property $\mathcal{P}_i$ by induction, for $0 \leq i \leq N-1$: Node $v_i$ eventually multicasts $(m,S_i)$.

\begin{itemize}
\item $\mathcal{P}_0$ is true, as $v_0 = v$ has multicast $(m,\o)$.

\item Let us suppose that $\mathcal{P}_i$ is true, with $i \leq N-2$.
Let $e$ be an execution where $v_i$ has multicast $(m,S_i)$.
Then, $v_{i+1}$ eventually receives $(m,S_i)$.
According to TRANSMISSION, as $card(S_i) \leq H-1$ and $v_i \notin S_i$,  $v_{i+1}$ eventually multicast $(m,S_{i+1})$. Therefore, $\mathcal{P}_{i+1}$ is true. 

\end{itemize}

So $\mathcal{P}_{N-1}$ is true and $v_{N-1}$ eventually multicasts $(m,S_{N-1})$.
Therefore, $p$ eventually receives $(m,S_{N-1})$.
According to TRANSMISSION, as $card(S_{N-1}) \leq H - 1$ and $v_{N-1} \notin S_{N-1}$, $(m,S_{N-1})$ is eventually added to $p.Trig$.
Thus, we eventually have $(m,q) \in p.Wait$, $(m,S_{N-1}) \in p.Trig$ and $q \notin S_{N-1}$.
So according to DECISION, $p$ eventually delivers $m$.
\end{proof}

\section{A Reliable Torus Network}

\label{sec_torus}

In this section, we refined the general conditions given in section~\ref{sec_prop} for the particular case of torus networks. Torus is good example of a multihop sparse topology, as every node has exactly four neighbors, and is sufficiently regular to permit analytical reasoning. 

\subsection{Preliminaries}

We first recall the definition of the torus topology:

\begin{definition}[Torus network]
\label{defgrid}
A $N \times N$ \emph{torus} network is a network such that:
\begin{itemize}
\item Each node has a unique identifier $(i,j)$ with
$1 \leq i \leq N$ and $1 \leq j \leq N$.
\item Two nodes $(i_{1},j_{1})$ and $(i_{2},j_{2})$ are neighbors if and only if one of these two conditions is satisfied:
\begin{itemize}
\item $i_{1} = i_{2}$ and $ | j_{1}-j_{2} |  = 1$ or $N$.
\item $j_{1} = j_{2}$ and $ | i_{1}-i_{2} | = 1$ or $N$.
\end{itemize}
\end{itemize}
\end{definition}

\paragraph{Tori \emph{vs} grids.}

If we remove the \emph{``or $N$''} from the previous definition, we obtain an arguably more realistic topology: the \emph{grid}. A grid network can easily be represented in a bidimensional space (see Figure~\ref{fig:grid}).

\begin{figure*}
\begin{center}
\includegraphics[width=5cm]{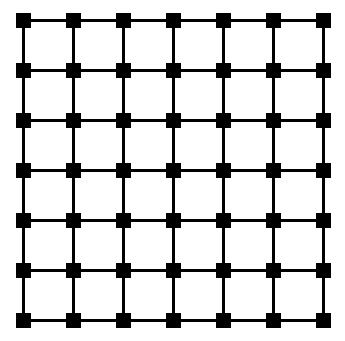}
\caption{Example of grid network: a $7 \times 7$ grid} 
\label{fig:grid}
\end{center}
\end{figure*}

However, no general condition on the distance between Byzantine nodes can guarantee reliable broadcast in the \emph{grid}. Indeed, let us suppose that the node $(2,2)$ is the source, and that the node $(1,2)$ is Byzantine. Then, the node $(1,1)$ has no way to know which node tells the truth between $(1,2)$ and $(2,1)$.

To avoid such border effects, we consider a \emph{torus} network in this part. The grid will be studied in Section~\ref{sec_exp}, with an experimental probabilistic study.

\subsection{A sufficient condition for reliable broadcast}

The main theorem of this section guarantees network safety, again in terms of spacing Byzantine nodes apart. This condition depends on the parameter $H$ of the protocol, and on the distance $D$ (see \ref{prelim_proto}). We also show that the condition on $D$ is tight for our protocol.

\begin{theorem} [Torus reliable broadcast]
Let $T$ be a torus network, and let the parameter of the protocol be $H = 2$.
If $D \geq 5$, all correct nodes eventually deliver the authentic message.
\label{thtorus}
\end{theorem}

\begin{proof}
According to Theorem~\ref{thsafe}, as $H = 2$ and $D \geq 5$, no correct node ever delivers a false message. In the sequel, the expression \emph{proof by exhaustion} designates a large number of trivial proofs that we do not detail, as they present no particular interest.

If the dimensions of the torus are $5 \times 5$ or less,
the proof of reliable broadcast is by exhaustion:
we consider each possible distribution of Byzantine nodes, and use Theorem~\ref{thcom} to show that all correct nodes eventually deliver the authentic message. Now, let use suppose that the dimensions of the torus are greater than $5 \times 5$.

Let $v$ be any correct node. Let $(u_1,\dots,u_n)$ be a path between the source $s$ and $v$. If this path is not correct, we can easily construct a correct path between $s$ and $v$.
Indeed, as $D \geq 5$, there exists a square correct path of $8$ hops around each Byzantine node. So, for each Byzantine node $u_i$ from the path, we replace $u_i$ by the correct path linking $u_{i-1}$ and $u_{i+1}$. Therefore, we can always construct a correct path $(p_1,\dots,p_n)$ between $s$ and $v$.

For a given node $p$, we call $G_{3 \times 3}(p)$ the $3 \times 3$ grid from which $p$ is the central node $(2,2)$, and $G_{5 \times 5}(p)$ the $5 \times 5$ grid from which $p$ is the central node $(3,3)$.
We want to prove the following property $\mathcal{P}_i$ by induction: all correct nodes of $G_{3 \times 3}(p_i)$ eventually deliver the authentic message.

\begin{itemize}
\item We prove $\mathcal{P}_1$ by exhaustion: 
we consider each possible distribution of Byzantine nodes in $G_{3 \times 3}(s)$ with $D \geq 5$, and use Theorem~\ref{thcom} to show that all correct nodes eventually deliver the authentic message.

\item Let us suppose that $\mathcal{P}_i$ is true.
$G_{3 \times 3}(p_{i+1})$ contains $p_i$ and at least two of its neighbors.
As $D \geq 5$, at least one on these neighbors $q$ is correct.
As $p_i$ and $q$ are also in $G_{3 \times 3}(p_{i})$, they eventually deliver the authentic message, according to $\mathcal{P}_i$.

\begin{itemize}

\item Let us suppose that there is no Byzantine node in $G_{3 \times 3}(p_{i+1})$.
Then, we prove $\mathcal{P}_{i+1}$ by exhaustion: 
we consider each possible distribution of Byzantine nodes in $G_{3 \times 3}(p_{i+1})$ with $D \geq 5$, and use Theorem~\ref{thcom} to show that all correct nodes eventually deliver the authentic message.

\item Let us suppose that there are some Byzantine node in $G_{3 \times 3}(p_{i+1})$.
According to our hypothesis, there is at most one Byzantine node $b$ in $G_{3 \times 3}(p_{i+1})$.
Then, all correct nodes of $G_{3 \times 3}(p_{i+1})$ are in $G_{5 \times 5}(b)$ -- so, in particular, $p_i$ and $q$.
As $D \geq 5$, $b$ is the only Byzantine node in $G_{5 \times 5}(b)$.
Then, we prove $\mathcal{P}_{i+1}$ by exhaustion: 
we consider each possible placement of $p_i$ and $q$ in $G_{5 \times 5}(b)$, and use Theorem~\ref{thcom} to show that all correct nodes of $G_{5 \times 5}(b)$ -- and thus, all correct nodes of $G_{3 \times 3}(p_{i+1})$ -- eventually deliver the authentic message.

\end{itemize}

\end{itemize}

So $\mathcal{P}_n$ is true, and $v = p_n$ eventually delivers the authentic message.

\end{proof}

As a complementary result, let us show that the bound $D \geq 5$ is tight for our protocol.

\begin{theorem} [Torus tight bounds]
If $D = 4$, some correct nodes may never deliver the authentic message.
\end{theorem}

\begin{proof}
Let $T$ be a $N \times N$ torus network, with $N \geq 8$.
Let us consider the example given in Figure~\ref{fig:critical}, where $D = 4$. In this figure, the central node $s$ is the source node. As they are direct neighbors of the source, the node of type $1$ eventually deliver the authentic message. However, the nodes of type $2$ never do so.

\begin{figure}
\begin{center}
\includegraphics[width=6cm]{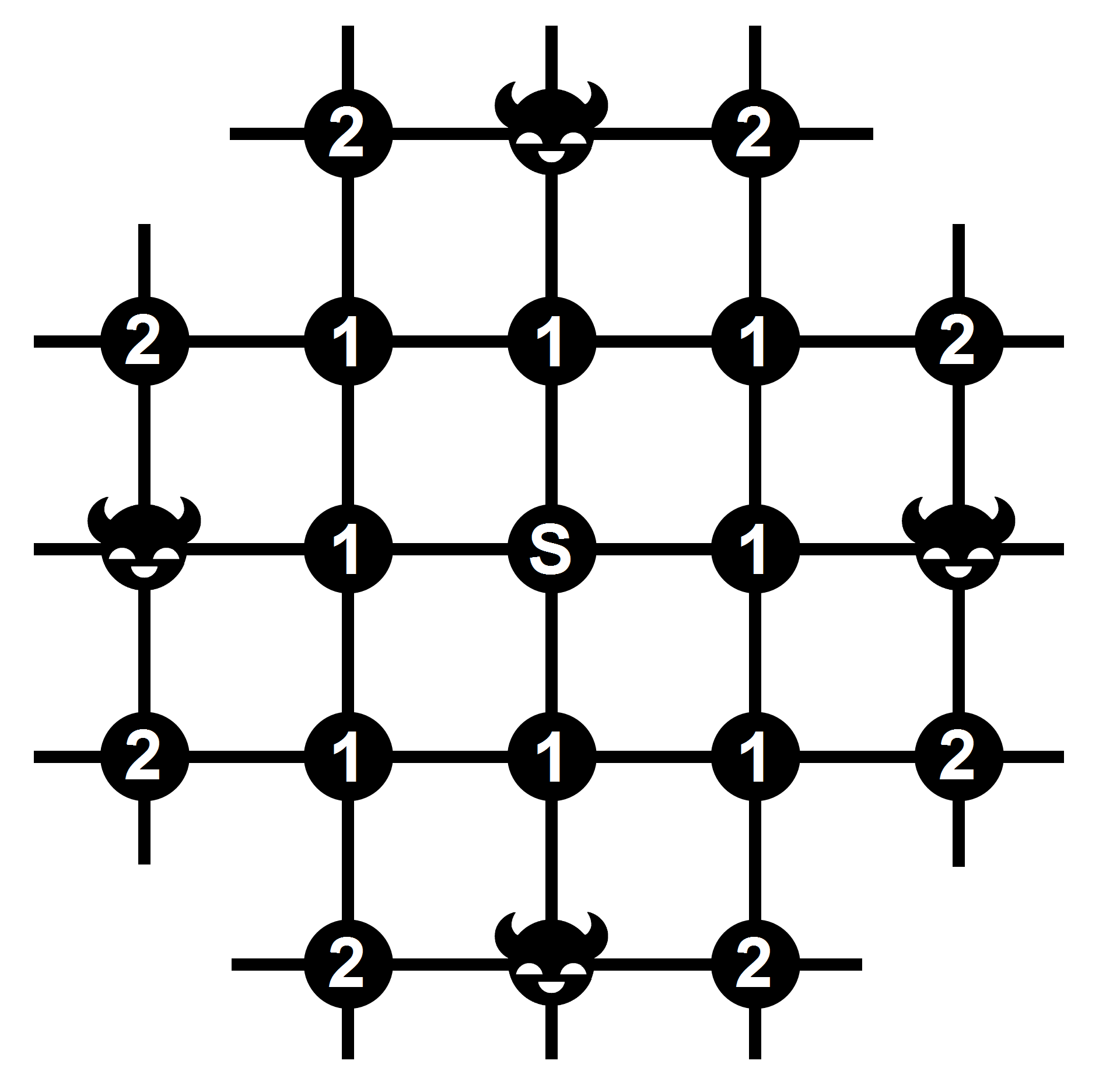}
\caption{Critical case in a torus network} 
\label{fig:critical}
\end{center}
\end{figure}

Indeed, let us consider a node $p$ of type $2$, and its neighbor $q$ of type $1$. To deliver the authentic message, $p$ needs to receive a trigger message from another node of type $1$, by a correct path of $H$ hops that does not contain $q$. But, as $H = 2$, such a path does not exist. Besides, we cannot take $H > 2$, as it would enable some correct nodes to deliver a false message, according to Theorem~\ref{thunsafe}. Therefore, the nodes of type $2$ -- and thus, the other correct nodes -- will never deliver the authentic message.

\end{proof}

Finally, let us discuss possible extensions to a grid-shaped network. We have seen that perfect reliable broadcast was impossible in a grid, due to border effects. However, it is actually possible in a sub-grid extracted from the grid.

More precisely, let $\mathcal{G}$ be a $N \times N$ grid, and $\mathcal{G'}$ a sub-grid containing all the nodes $(i,j)$ of $\mathcal{G}$ such that $4 \leq i \leq N-4$ and $4 \leq j \leq N-4$. Then, the proof of Theorem~\ref{thtorus} is also valid for $\mathcal{G'}$.

It is also the case if we consider \emph{any} particular node in an infinite grid (but not \emph{all} nodes). In other words, a given correct node eventually delivers the authentic message, even if the notion of perfect reliable broadcast does not make sense in an infinite network.

\section{Experimental evaluation}

\label{sec_exp}

In this section, we target quantitative Byzantine resilience evaluation when considering the case of randomly distributed Byzantine failures. We first give a methodology to estimate the number of Byzantine failures that a particular network can tolerate for a given \emph{probabilistic} guarantee. Then, we present experimental results for a \emph{grid} topology.

Notice that only the \emph{placement} of Byzantine failures is probabilistic: once this placement is determined, we must assume that the Byzantine nodes adopt the worst possible strategy, and that the worst possible execution may occur.

\subsection{Methodology}

Let $n_B$ be the number of Byzantine failures, randomly distributed on the network (the distribution is supposed to be uniform). We would like to evaluate the \emph{probability} $P(n_B)$, for a correct node, to deliver the authentic message. 
For this purpose, we use a Monte-carlo method:

\begin{itemize}
\item We generate several random distributions of $n_B$ Byzantine failures.
\item For each distribution, we randomly choose a source node $s$ and a correct node $v$. Then, we use Theorem~\ref{thcom} to construct a \emph{reliable node set} (see Definition~\ref{defrel}). If $v$ is in the reliable node set, it eventually delivers  the authentic message, and the simulation is a success -- else, it is a failure.
\item With a large number of simulations, the fraction of successes will approximate $P(n_B)$.
\end{itemize}

More precisely, we approximate a \emph{lower bound} of $P(n_B)$, as the reliable node set constructed in not necessarily the best. Therefore, we can determine a maximal number of Byzantine failures that can be tolerated for a given guarantee (for instance: $P(n_B) \geq 0.99$).

\subsection{Results}

We run simulations on $N \times N$ grid networks, with a parameter $H = 2$ for the protocol. The results are presented in Figure~\ref{simu}.

\begin{figure}
\begin{center}
\includegraphics[width=\textwidth]{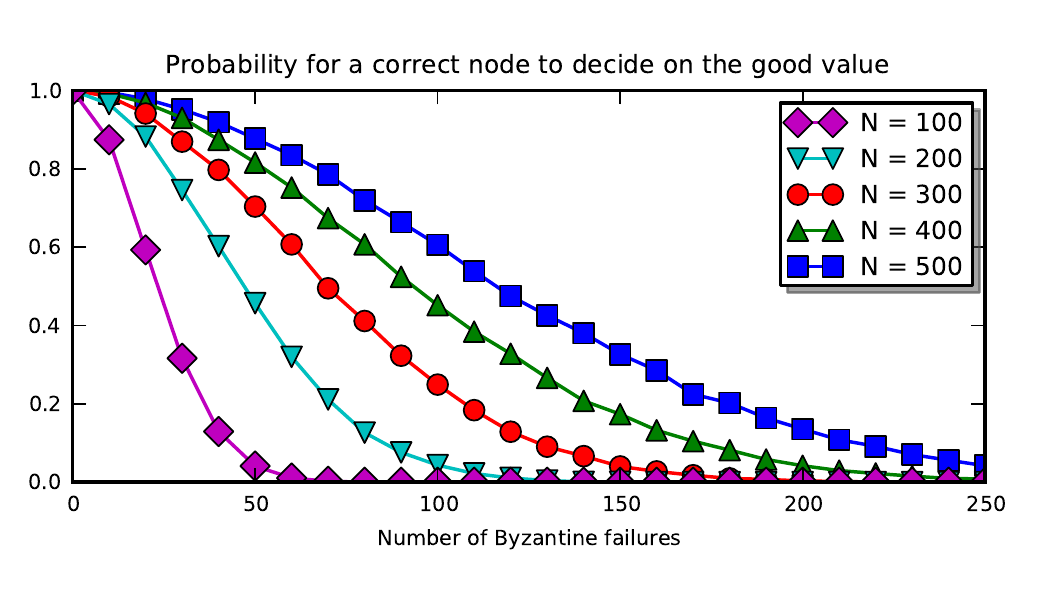}
\caption{Experimental evaluations on $N \times N$ grid networks} 
\label{simu}
\end{center}
\end{figure}

As expected, a larger grid can tolerate more Byzantine failures, as they are more likely to be sufficiently spaced.

To our knowledge, the only existing protocol working on such a sparse topology -- without specific initialization of the nodes -- is Explorer \cite{NT09j}. This protocol consists in a voting system on \emph{node-disjoint} paths between the source and the peers. However, as a node has at most $4$ neighbors, $2$ Byzantine failures can prevent \emph{any} correct node to deliver the authentic message.
Therefore, no guarantee can be given for more than $1$ Byzantine failure.

As in~\cite{CtrZ}, we could have modified Explorer and forced it to use predetermined paths on the grid. However, this would require global topology knowledge. More precisely, in order to use such a tweaked version of Explorer, a node must know its position on the grid and, for a given neighbor, whether it is its upper, lower, left or right neighbor. Those assumptions are not required with our protocol.

On this grid topology, our protocol enables to tolerate \emph{more} than $1$ Byzantine failure with a good probability. For instance, for $N = 500$, we can tolerate up to $14$ Byzantine failures with $P(n_B) \geq 0.99$ (see Figure~\ref{simu}).

\section{Conclusion}

In this paper, we proposed a Byzantine-resilient broadcast protocol for loosely connected networks that does not require any specific initialization of the nodes, nor global topology knowledge. We gave a methodology to construct a reliable node set, then sufficient conditions for perfect reliable broadcast in a sparse topology: the torus. Finally, we presented a methodology to determine the number on randomly distributed Byzantine failures that a network can hold.

Several interesting open questions remain. First, we have the strong intuition that the condition proved on the torus could be generalized to \emph{any} network topology. Another challenging problem is to obtain theoretical probabilistic guarantees, based on global network parameters such as diameter, node degree or connectivity. Third, the tradeoff between global knowledge and the number of Byzantine nodes that can be tolerated requires further attention.

\bibliographystyle{plain}
\bibliography{biblio}

\end{document}